\documentclass[12pt]{article}
\usepackage{hyperref}
\usepackage{color}
\usepackage{amsthm}
\usepackage{amssymb}
\usepackage{amsmath}
\usepackage{epsfig}
\usepackage{epstopdf}
\usepackage{graphicx}
\usepackage{indentfirst}
\usepackage{multirow}
\usepackage{array,arydshln}
\usepackage{makecell}
\usepackage{subfigure}

\renewcommand{\baselinestretch}{1.6}

\def\MCD{{\mbox{\small MCD}}}

\def\OA{{\mbox{\small OA}}}

\def\OA{{\mbox{\small OA}}}

\def\LHD{{\mbox{\small LHD}}}
\def\MCD{{\mbox{\small MCD}}}

\renewcommand{\arraystretch}{1.2}
\newtheorem{lem}{Lemma}

\newtheorem{theorem}{Theorem}

\newtheorem{proposition}{Proposition}
\newtheorem{example}{Example}
\newtheorem{defi}{Definition}

{\begin{list}{}%
         {\setlength{\leftmargin}{#1}}%
         \item[]%
}
{\end{list}}

\date{}

\begin{document}

\title{\bf Construction of Marginally Coupled Designs by Subspace Theory}

\author{\small Yuanzhen He$^1$, C. Devon Lin$^{2}$ and Fasheng Sun$^3$\footnote{Corresponding author; E-mail: sunfs359@nenu.edu.cn.} \\
\small \em $^1$Beijing Normal University, $^2$Queen's University, $^{3}$Northeast Normal University}

\maketitle
\begin{abstract}
Recent researches on designs for computer experiments with both qualitative and quantitative factors have advocated the use of marginally coupled designs.
This paper proposes a general method of constructing such designs for which the designs for qualitative factors are multi-level orthogonal arrays and
the designs for quantitative factors are  Latin hypercubes with desirable space-filling properties.   Two cases
 are introduced for which we can obtain the guaranteed low-dimensional space-filling property
 for quantitative factors. Theoretical results on the proposed constructions are derived. For practical use, some constructed designs for three-level qualitative factors are tabulated.


\end{abstract}

\section{Introduction}

Computer experiments with both qualitative and quantitative variables are becoming increasingly common
(see, for example, Rawlingson et al., 2006; Qian, Wu and Wu, 2008; Han et al., 2009; Zhou, Qian and Zhou, 2011; Deng et al., 2017). Extensive studies have been devoted to design and modeling of such experiments.
This article focuses on a particular class of designs, namely, {\em marginally  coupled designs}, which have been argued to be a cost-effective design choice (Deng, Hung and Lin, 2015). The goal here is to propose a general method for constructing marginally coupled designs when the design for qualitative variables is a multi-level orthogonal array.

The first systematical plan  to accommodate computer experiments with both qualitative and quantitative variables is sliced Latin hypercube designs
proposed by Qian and Wu (2009). In such a design, for each level combination of the qualitative factors,
the corresponding design for the quantitative factor is a small Latin hypercube (McKay, Beckman and Conover, 1979).
The run size of a sliced Latin hypercube design increases dramatically  with the number of the qualitative factors.
To accommodate a large number of qualitative factors with an economical run size, Deng, Hung and Lin (2005) introduced marginally coupled designs which possess the property that with respect to each level of each qualitative variable, the corresponding design for quantitative variables is a sliced Latin hypercube design.  Other enhancements of sliced Latin hypercubes include multi-layer sliced Latin hypercube designs (Xie et al., 2014), clustered-sliced Latin hypercube designs (Huang et al., 2016), bi-directional sliced Latin hypercube designs (Zhou et al., 2016).


Since being introduced by Deng, Hung and Lin (2015), there have been two developments of marginally coupled designs, due to He, Lin and Sun (2017) and He et al. (2017), respectively.  Comparing with the original work, both developments provide designs for quantitative factors without clustered points, thereby improving the space-filling property which refers to spreading out points in the design region as evenly as possible (Lin and Tang, 2015).
He, Lin and Sun (2017) constructs marginally coupled designs of $s^u$ runs that can accommodate $(s+1-k)s^{u-2}$ qualitative factors and $k$ quantitative factors for a prime power $s$ and $1 \leq k < s+1$.
The drawback of this method is when $s=2$, the corresponding designs can accommodate only up to $3$ quantitative factors.
 He et al. (2017)  addressed this issue and introduced a method for constructing marginally coupled designs of $2^u$ runs
for $2^{u_1-1}$ qualitative factors of two levels and up to $2^{u-u_1}$ quantitative factors, where $1 \leq u_1 \leq u$.

The paper aims to construct marginally coupled designs of $s^u$ runs in which designs for qualitative factors are $s$-level orthogonal arrays  for a prime power $s$ and any positive integer $u$. The primary technique in the proposed construction is the subspace theory of Galois field $GF(s^u)$. Although such a technique was used in the constructions in He et al. (2017) for $s=2$, it is not trivial to generalize their constructions for any prime power $s$. Extra care must be taken in the generalization. The other contribution of this article is to   introduce two cases
for which guaranteed low-dimensional space-filling property  for quantitative factors can be obtained.
For example, for $s=2$, the designs of $2^u$ runs for quantitative factors achieve stratification on a $2 \times 2 \times 2$ grid of any three dimensions.

The remainder is arranged as follows. Section 2 introduces background and preliminary results.
New constructions and the associated theoretical results are presented in Section 3.
Section 4 tabulates the designs with three-level qualitative factors.
The space-filling property of the newly constructed designs is discussed in Section 5,
and the last section concludes the paper. All the proofs are relegated to Appendix.

\section{Background and Preliminary Results}

\subsection{Background}

A matrix of size $n\times m$, where the $j$th column has $s_j$ levels $0,\ldots,s_j-1$, is called an orthogonal array of strength $t$,
if for any $n\times t$ sub-array, all possible level combinations appear equally often.
It is denoted by $\OA(n, s_1\cdots s_m, t)$   and  the simplified notation $\OA(n, s^{u_1}_1s^{u_2}_2\cdots s^{u_k}_k, t)$ will be used if the first $u_1$ columns have $s_1$ levels, the next $u_2$ columns  have $s_2$ levels, and so on. If $s_1=\cdots=s_m=s$,
it is shortened as $\OA(n, m, s, t)$. If all rows of an $\OA(n, m, s, t)$ can form a vector space,
it is called a linear orthogonal array (Hedayat, Sloane and Stufken, 1999).
For a prime power $s$, let $GF(s)=\{\alpha_0, \alpha_1, \ldots, \alpha_{s-1}\}$ be a Galois field of order $s$,
where $\alpha_0=0$ and $\alpha_1=1$. Throughout this paper, unless otherwise specified, entries of any $s$-level array  are from $GF(s)$.
For a set $S$, {\bf $|S|$} represents the number of elements in $S$.

A Latin hypercube is an $n\times k$ matrix each column of which is a random permutation of $n$ equally spaced levels (McKay, Beckman and Conover, 1979).
In this article, these $n$ levels are represented by $0, \ldots, n-1$, and a Latin hypercube  of $n$ runs for $k$ factors is denoted by $\LHD(n, k)$.
A special type of Latin hypercubes is a {\em cascading Latin hypercube} for which with $n=n_1n_2$ points and levels $(n_1, n_2)$ is an $n_2$-point Latin hypercube  about each point in the $n_1$-point Latin hypercube (Handcock, 1991).
Latin hypercubes can be obtained from orthogonal arrays.
Given an $\OA(n, m, s, t)$, replace the $r=n/s$ positions having level $i$ by a random permutation
of $\{ ir,  \ldots, (i+1)r-1\}$, for $i=0,\ldots,s-1$.
The resulting design achieves $t$-dimensional stratification, and is called an
orthogonal array-based Latin hypercube (Tang, 1993).
This approach is referred to as the {\it level replacement-based Latin hypercube} approach.

Let $D_1$ be an $\OA(n, m, s, 2)$ and $D_2$ be an $\LHD(n, k)$. Design $D=(D_1, D_2)$ is called
a {\em marginally coupled design}, denoted by $\MCD(D_1, D_2)$, if for each level of every column of $D_1$, the corresponding rows
in $D_2$ have the property that when projected onto each column, the resulting entries consist of exactly one level from
each of the $n/s$ equally-spaced intervals $\{[0,s-1],[s,2s-1],\ldots,[n-s,n-1]\}.$  As a space-filling design is generally sought, a $D_2$ in which the whole design or any of its column-wise projections has clustered points shall be avoided.
We define a Latin hypercube $D_2$ to be {\em non-cascading}  if, when projected onto any two distinct columns of $D_2$, the resulting design is not a cascading Latin hypercube of levels $(s,n/s)$.

To study the existence of $\MCD(D_1,D_2)$'s, He, Lin and Sun (2017) defined the matrix
$\tilde{D}_2$ based on $D_2$.  Let $d_{2,ij}$     be the ($i,j)$th entry of $D_2$. The ($i,j$)th entry $\tilde{d}_{2,ij}$ is given by
\begin{equation}\label{eq:tildeD2}
\tilde{d}_{2,ij} =\Big \lfloor d_{2,ij}/s\Big\rfloor,  \ i=1,\ldots, n \ \hbox{and} \  j=1,\ldots,k,
\end{equation}
\noindent where  $\lfloor x\rfloor$ denotes the greatest integer less than or equal to $x$. The operator in (\ref{eq:tildeD2})
scales the levels in the interval $[0,s-1]$ to level 0, the levels in the interval $[s,2s-1]$ to level 1, and so on. Thus, the levels in $\tilde{D}_2$ are  $\{0, 1, \ldots, n/s-1\}$.
On the other hand, design $D_2$ can be obtained from $\tilde{D}_2$  via the {\it level
replacement-based Latin hypercube} approach.
Lemma \ref{lem:D1-D2-condition}  given by  He, Lin and Sun (2017) provides a necessary and sufficient condition
for the existence of an $\MCD(D_1,D_2)$ when $D_1$ is an $s$-level orthogonal array.

\begin{lem}\label{lem:D1-D2-condition}
Given that $D_1$ is an $\OA(n, m, s, 2)$, $D_2$ is an $\LHD(n, k)$ and $\tilde{D}_2$ is defined via (\ref{eq:tildeD2}),
then $(D_1, D_2)$ is a marginally coupled design if and only if for $j=1,\ldots,k$,
$(D_1, {\bf d}_j)$ is an $\OA(n, s^m(n/s), 2)$, where ${\bf d}_j$ is the $j$th column  of $\tilde{D}_2$.
\end{lem}

In addition to conveniently study the existence of marginally coupled designs, the definition of $\tilde{D}_2$
allows us to determine whether or not $D_2$ is {\em non-cascading}.  By definition, a Latin hypercube $D_2$ is  {\em non-cascading} if any two distinct columns of the corresponding $\tilde{D}_2$
cannot be transformed to each other by level permutations.

\subsection{Preliminary results}


This subsection presents a result that is the cornerstone of the proposed general construction in next section. Although the result itself is trivial, it is important to review the notation, concepts and existing results to help understand the later development. An example is also given to facilitate the understanding. Suppose that we wish to construct  an $\MCD(D_1,D_2)$ with $D_1 = \OA(s^u, m,s,2)$ and $D_2 = \LHD(s^u, k)$. Lemma \ref{lem:D1-D2-condition} indicates that it is equivalent to construct $D_1=( {\bf a}_1, \ldots, {\bf a}_m)$  and  $\tilde D_2=({\bf d}_1, \ldots, {\bf d}_k) = \OA(s^u,k, s^{u-1},1)$ such that  $({\bf d}_j, {\bf a}_i)=\OA(s^u, s^{u-1}\times s,2)$ (Here $s^{u-1}\times s$ means ${\bf d}_j$ has $s^{u-1}$ levels, and ${\bf a}_i$ has $s$ levels) and any distinct two columns ${\bf d}_i$ and ${\bf d}_j$ cannot be transformed to each other by level permutations. This subsection focuses on a construction of an $\OA(s^u, s^{u-1} \times s, 2)$.

First, we review the connection between  an $s^{u-1}$-level column and a $(u-1)$-dimensional subspace of $GF(s^w)$, where $w\geq u-1$.  To see this, note that an $s^{u-1}$-level column can be generated  by  choosing a subarray $A_0=\OA(s^w, u-1,s, u-1)$  from a linear $\OA(s^{w}, m, s, 2)$, say $A$, and substituting each level combination of these columns
by a unique level of $\{0,1, \ldots, s^{u-1}-1\}$ in some manner. This procedure is known as the {\em method of replacement} (Wu and Hamada, 2011).  One method  to achieve the substitution is $A_0\cdot(s^{u-2}, \ldots, s, 1)^T$,     where the superscript $T$ represents the transpose
of a matrix or a vector; this is exactly what we adopt in this paper.
The $A_0$, consisting of $u-1$ independent columns, can also be generated using all linear combinations of rows of a  $w\times (u-1)$ matrix $G$, called the {\it generator matrix} of $A_0$ (Hedayat, Sloane and Stufken, 1999).  In addition, all linear combinations of columns of $G$ form a $(u-1)$-dimensional vector subspace of $GF(s^w)$.
Therefore, an $s^{u-1}$-level column corresponds
to one $(u-1)$-dimensional subspace of $GF(s^w)$, where $w\geq u-1$.

Consider the case of $w=u$. Let $S_u$ consist of $s$-level column vectors of length $u$,
then all of its column vectors form a space of dimension $u$.
For the detail of vector spaces, refer to Horn and Johnson (2015).
For two column vectors ${\bf x}, {\bf y}\in S_u$, if ${\bf x}^T{\bf y}=0$ in $GF(s)$,
  they are said to be orthogonal.
For a nonzero element ${\bf x}\in S_u$, define
\begin{equation}\label{def:O(x)}
O({\bf x})=\{{\bf y}\in S_u\ | \ {\bf y}^T{\bf x}=0\}.
\end{equation}
It can be seen that $O({\bf x})$ is a $(u-1)$-dimensional subspace of $S_u$.

Let $G({\bf x})$ be a $u\times (u-1)$ matrix consisting of
$u-1$ independent columns of $O({\bf x})$. For a vector from $S_u\setminus O({\bf x})$, say ${\bf z}$,
all linear combinations of rows of the matrix $(G({\bf x}), {\bf z})$
can generate an $s^u\times u$ matrix. For ease of presentation, the first $u-1$ columns and the last column of the resulting matrix are denoted by $A({\bf x})$ and ${\bf a}$, respectively.
Applying the {\em method of replacement} to $A({\bf x})$ yields  an $s^{u-1}$-level vector, say ${\bf d}$.
Lemma~\ref{lemma:basic-idea} indicates that the ${\bf d}$ and ${\bf a}$ are orthogonal.

\begin{lem}\label{lemma:basic-idea}
For ${\bf d}$ and ${\bf a}$ constructed above, we have that $({\bf d}, {\bf a})$ is an $OA(s^u, s^{u-1}\times s, 2)$.
\end{lem}

\begin{example}
For $s=u=3$, we have $GF(3)=\{0, 1, 2\}$ and $S_3=\{(x_1, x_2, x_3)^T \mid x_i\in GF(3), i=1, 2,3\}$. Consider ${\bf x}=(1, 2, 0)^T$, and we have
\begin{equation*}
O({\bf x}) =
\left(\begin{array}{ccc ccc ccc}
  0 & 0 & 0 & 1 & 1 & 1 & 2 & 2 & 2 \\[-6pt]
  0 & 0 & 0 & 1 & 1 & 1 & 2 & 2 & 2 \\[-6pt]
  0 & 1 & 2 & 0 & 1 & 2 & 0 & 1 & 2
  \end{array}
  \right),
\end{equation*}
\noindent and the dimension of $O({\bf x})$ is 2. Choose two independent columns $(0, 0, 1)^T$ and $(1, 1, 0)^T$ from $O({\bf x})$, and column-combining them gives
$G({\bf x})$.
For ${\bf z}=(1, 2, 0)^T\in S_3\setminus O({\bf x})$,
$(G({\bf x}), {\bf z} )$ generates a $27\times 3$ matrix $(A({\bf x}), {\bf a})$, whose transpose
is as follows
\[
\left(\begin{array}{ccccccccc ccccccccc ccccccccc}
0&1&2&0&1&2&0&1&2&0&1&2&0&1&2&0&1&2&0&1&2&0&1&2&0&1&2\\[-6pt]
0&0&0&1&1&1&2&2&2&1&1&1&2&2&2&0&0&0&2&2&2&0&0&0&1&1&1\\[-6pt]
0&0&0&2&2&2&1&1&1&1&1&1&0&0&0&2&2&2&2&2&2&1&1&1&0&0&0
 \end{array}
 \right).
\]
By the method of replacement, let ${\bf d}=A({\bf x})\cdot(3, 1)^T$. Then $({\bf d}, {\bf a})$ is an $OA(27, 9\times 3, 2)$ whose
transpose is
\[
\left(\begin{array}{ccccccccc ccccccccc ccccccccc}
0&3&6&1&4&7&2&5&8&1&4&7&2&5&8&0&3&6&2&5&8&0&3&6&1&4&7\\[-6pt]
0&0&0&2&2&2&1&1&1&1&1&1&0&0&0&2&2&2&2&2&2&1&1&1&0&0&0
\end{array}\right).
\]
\end{example}

\section{Construction}\label{sec:2}

This section introduces a general construction and a subspace construction for marginally coupled designs using a set of vectors from $S_u$. For each construction, a necessary condition for the set of vectors is given. For the given
design parameters $s, u, u_1$, two constructions provide marginally coupled designs with different numbers of qualitative factors and quantitative factors. The key results are summarized in Theorems 1 and 2.

In the following constructions, when choosing nonzero vectors ${\bf x}, {\bf y}$ from $S_u$ to construct
orthogonal arrays or to construct $(u-1)$-dimensional subspaces $O({\bf x})$ and $O({\bf y})$,
we require ${\bf x}\neq \alpha {\bf y}$ for any $\alpha \in GF(s)$. This is because if ${\bf x}=\alpha{\bf y}$
for some $\alpha\in GF(s)$, ${\bf x}$ and ${\bf y}$ generate the columns representing the same factor, and
$O({\bf x})$ and $O({\bf y})$ actually represent the same $(u-1)$-dimensional subspace.

\subsection{General construction}

Suppose we choose $m+k$ vectors ${\bf z}_1, \ldots, {\bf z}_m, {\bf x}_1, \ldots, {\bf x}_k$ from $S_u$,
such that ${\bf z}_i$ is not in any of $O({\bf x}_j)$. We propose the following three-step construction.

\begin{itemize}
  \item[Step 1.]  Obtain  $D_1=({\bf a}_1, \ldots, {\bf a}_m)$ by taking all linear combinations of the rows of $({\bf z}_1, \ldots, {\bf z}_m)$, where ${\bf a}_i$ is the $i$th column of $D_1$;

  \item[Step 2.] For each ${\bf x}_j$, choose $u-1$ independent columns from  $O({\bf x}_j)$ in (\ref{def:O(x)}) to form a generator matrix $G({\bf x}_j)$. Obtain  $A({\bf x}_j)$ by taking all linear combinations of the rows of $G({\bf x}_j)$.  Apply the
                 {\em method of replacement} to obtain an $s^{u-1}$-level column vector ${\bf d}_j$  from $A({\bf x}_j)$. Denote the resulting design by $\tilde D_2=({\bf d}_1, \ldots, {\bf d}_k)$;

  \item[Step 3.]  Obtain $D_2$ from $\tilde D_2$ via the {\it level replacement-based Latin hypercube} approach.
\end{itemize}

The method of obtaining ${\bf d}_j$ and $ {\bf a}_i$ in Steps 1 and 2 in the general construction are essentially the construction in Section 2.2 and thus by  Lemma \ref{lemma:basic-idea}, $({\bf d}_j, {\bf a}_i)$ is an $OA(s^u, s^{u-1}\times s, 2)$.  In addition, $D_1$ is an $\OA(s^u, m, s, 2)$  and $D_2$ is an  $\LHD(s^u,k)$. Therefore,  the $(D_1, D_2)$ is a marginally coupled design. The condition of the construction is to have ${\bf z}_i$ not in any of $O({\bf x}_j)$.
To find such  ${\bf z}_i$'s and ${\bf x}_j$'s, we consider the set of vectors $\{{\bf e}_1, \ldots, {\bf e}_{u_1}\}\subset S_u$,
where ${\bf e}_i$ is a vector of $S_u$ with the $i$th entry equal to 1 and the other entries equal to 0, and
$1\leq u_1\leq u$. We further define
\begin{equation}\label{def:mathcal{A}}
\mathcal{A}=\{ {\bf x}\in S_u\setminus (\cup_{i=1}^{u_1} O({\bf e}_i)) \mid  \mbox{the first entry of ${\bf x}$ is 1}\},
\end{equation}
where $O(\cdot)$ is defined in (\ref{def:O(x)}). The main result of using $\mathcal{A}$ and ${\bf e}_i$'s to construct
$\MCD(D_1,D_2)$'s is provided in Theorem \ref{thm-simple-construction}. Before presenting the theorem, we describe a result which counts the number of vectors in $\mathcal{A}$.

\begin{lem}\label{lemma:general-union-O(e)}
There are $n_A=(s-1)^{u_1-1}s^{u-u_1}$ column vectors in $\mathcal{A}$ in (\ref{def:mathcal{A}}).
\end{lem}

The value of $n_A$ is the number of columns in $D_1$ or $D_2$, as revealed in Theorem \ref{thm-simple-construction}.

\begin{theorem}\label{thm-simple-construction}
For  $\{{\bf e}_1,\ldots, {\bf e}_{u_1}\}$ defined above, $\mathcal{A}$ in (\ref{def:mathcal{A}}) and $n_A$ in Lemma \ref{lemma:general-union-O(e)}, if in the general construction we
\begin{itemize}
  \item[(i)] choose ${\bf z}_i={\bf e}_i$ and ${\bf x}_j\in \mathcal{A}$
             for $1\leq i\leq u_1$ and $1\leq j\leq n_A$, an $\text{MCD}(D_1, D_2)$ with $\ D_1=\OA(s^u, u_1, s, u_1), D_2=\LHD(s^u, n_A)$
             can be obtained, or,
  \item[(ii)] choose ${\bf z}_i\in \mathcal{A}$ and ${\bf x}_j={\bf e}_j$ for $1\leq i\leq n_A$ and $1\leq j\leq u_1$, an
              $\text{MCD}(D_1, D_2)$ with $D_1=\OA(s^u, n_A, s, 2), D_2=\LHD(s^u, u_1)$ can be obtained,
\end{itemize}
where both $D_2$'s are non-cascading Latin hypercubes.
\end{theorem}

The design $D_1$(or $D_2$) in Theorem \ref{thm-simple-construction} (i)  (or (ii)) can only accommodate $u_1\leq u$ columns.
A natural question is whether or not more columns in $D_1$ (or $D_2$) can be constructed. The answer is positive for $s=2$ as shown in He et al. (2017) by choosing some linear
combinations of $\{{\bf e}_1, \ldots, {\bf e}_{u_1}\}$ besides themselves for
${\bf z}_i$'s (or ${\bf x}_j$'s). For $s>2$, the answer is still positive, however, there is a price to pay. That is, when more columns of $D_1$ than those in Theorem~\ref{thm-simple-construction} are constructed using some linear combinations of
$\{{\bf e}_1, \ldots, {\bf e}_{u_1}\}$ in addition to themselves, the number of columns in $D_2$ will be less than that in Theorem~\ref{thm-simple-construction}. The reason for paying such cost is quantified in Proposition~\ref{thm-combination is impossible}.

\begin{proposition}\label{thm-combination is impossible}
For $s>2$ and the set $\{{\bf e}_1,\ldots, {\bf e}_{u_1}\}$ defined above,
let ${\bf z}=\sum_{i=1}^{u_1}\lambda_i {\bf e}_i$ with at least two nonzero coefficients,
where $\lambda_i\in GF(s)$. For such ${\bf z}$'s and $\mathcal{A}$ in (\ref{def:mathcal{A}}), there exists a column vector ${\bf x}\in \mathcal{A}$,
such that ${\bf z}\in O({\bf x})$.
\end{proposition}

Proposition \ref{thm-combination is impossible} shows that, when $s>2$,
except $\{\alpha{\bf e}_i \mid \alpha\in GF(s)\setminus\{0\}, i=1, \ldots,  u_1\}$,
for any of their other combinations, say ${\bf z}$, it is impossible that ${\bf z}$ is not in $O({\bf x})$ for all ${\bf x}\in \mathcal{A}$. This means if adding ${\bf z}$ for constructing one more column for $D_1$, not all the columns in $\mathcal{A}$ can be used for   constructing columns for  $D_2$.  As a compromise, after adding more combinations of $\{{\bf e}_1, \ldots, {\bf e}_{u_1}\}$ for $D_1$,
we use a subset $\{{\bf x}_1, \ldots, {\bf x}_k\}\subset\mathcal{A}$ to construct $(u-1)$-dimensional subspaces $\{O({\bf x}_1), \ldots, O({\bf x}_k)\}$, where $k<n_A$. Next section discusses an approach to find such a subset.

\subsection{Subspace construction}

This subsection introduces an approach to find a proper subset $\{{\bf x}_1, \ldots, {\bf x}_k\}\subset\mathcal{A}$
and judiciously select some linear combinations
${\bf z}=\lambda_1{\bf e}_1+ \cdots + \lambda_{u_1}{\bf e}_{u_1}$, with $\lambda_j\in GF(s)$,
such that ${\bf z}\in S_u\setminus (\cup_{i=1}^kO({\bf x}_i))$.

One building block of the proposed approach is some disjoint groups of
$\mathcal{A}$. To partition $\mathcal{A}$ into different groups, note that for $1\leq j\leq u_1$, the last $u-u_1$ entries of ${\bf e}_j$
are zeros and  thus the first  $u_1$ entries of ${\bf z}$ and ${\bf x}_i$  determine whether or not ${\bf z}$ is orthogonal to ${\bf x}_i$.
In light of this observation, the partition of $\mathcal{A}$ is based on  the distinct values of the first $u_1$ entries of vectors in $\mathcal{A}$.   The proof of Lemma \ref{lemma:general-union-O(e)} reveals that the first $u_1$ entries of
${\bf x}\in \mathcal{A}$ can take $n_B=(s-1)^{u_1-1}$  distinct values, say $\{(1, b_{i2}, \ldots, b_{iu_1})\mid i=1, \ldots, n_B\}$. Let ${\bf b}_i=(1, b_{i2}, \ldots, b_{iu_1}, 0,\ldots, 0)^T$, and
define $\mathcal{A}_i$ to be the subset of $\mathcal{A}$ whose column vectors
have the same first $u_1$ entries  as those of ${\bf b}_i$. It shall be noted that $|\mathcal{A}_i|=s^{u-u_1}$ and  $\mathcal{A}_i$'s form a disjoint partition of  $\mathcal{A}$. That is,
\begin{equation}\nonumber
\mathcal{A}=\cup_{i=1}^{n_B}\mathcal{A}_i.
\end{equation}

The other building block  is a set of $\overline E_i$'s defined as follows.
Let $E=\{ \sum_{j=1}^{u_1}\lambda_j {\bf e}_j \mid \lambda_j\in GF(s)\}$
consist of all linear combinations of ${\bf e}_1, \ldots, {\bf e}_{u_1}$.
For fixed $i$, ${\bf b}_i$ and $\mathcal{A}_i$, $1\leq i\leq n_B$, define
\begin{eqnarray}\nonumber
E_i=\{\ {\bf z}\in E \mid {\bf z}^T{\bf b}_i=0 \ \} \ \mbox{and} \ \overline E_i=E\setminus E_i.
\end{eqnarray}
If ${\bf z}\in \overline E_i$, then ${\bf z}\notin O({\bf b}_i)$,
which implies ${\bf z}\notin O({\bf x})$ for all ${\bf x}\in \mathcal{A}_i$ since
the last $u-u_1$ entries of ${\bf z}$ are zeros. This leads to
Lemma \ref{lemma:Ei-and-Mi}.
\begin{lem}\label{lemma:Ei-and-Mi}
For $1\leq v\leq n_B$, any ${\bf z}\in \cap_{i=1}^v \overline E_{i}$ and
any ${\bf x}\in \cup_{i=1}^v \mathcal{A}_{i}$, we have ${\bf z}\notin O({\bf x})$.
\end{lem}

Lemma \ref{lemma:Ei-and-Mi} is useful because it provides
$\{{\bf z}_i\}$'s and $\{{\bf x}_j\}$'s required by the general construction in Section 3.1. That is, one can choose ${\bf z}_i$ from $\cap_{i=1}^v \overline E_{i}$,
and ${\bf x}_j$ from $\cup_{i=1}^v \mathcal{A}_{i}$, that is exactly the
method Theorem \ref{theorem:subspace-construction} adopts.

So far, it remains to resolve the question that what the elements are in
$\cap_{i=1}^{v} \overline E_{i}$ for $1\leq v \leq n_B$.
The answer is not difficult for $v=1$, and that
for $v=n_B$ can be found in Proposition~\ref{prop-the-bound-for-intersectionset}
in Appendix for interested readers.
For $1< v< n_B$, the explicit form for elements in $\cap_{i=1}^{v}\overline E_{i}$
depends on the specific sets $\overline E_{1}, \ldots, \overline E_{v}$.   Thus, we cannot express the elements
in $\cap_{i=1}^{v}\overline E_{i}$  using a general form. However, we are able to compute the number of elements in $\cap_{i=1}^{v}\overline E_{i}$
  for some cases.
Theorem \ref{theorem:subspace-construction}
 shows that this number is closely related to the number of variables in the marginally
coupled design. In practice, experimenters also hope to
know the number in advance, as it can help them determine which marginally coupled design to choose given the numbers of qualitative and quantitative variables in the experiment.
Proposition \ref{prop:s-level} below provides the number,  {\bf $|\cap_{i=1}^{v}\overline E_{i}|$}, in some circumstances.

\begin{proposition}\label{prop:s-level}
For $\{{\bf b}_1, \ldots, {\bf b}_{n_B}\}$ defined above,
suppose that there exists a subset $\{{\bf b}_{i_1}, \ldots, {\bf b}_{i_{n^*}}\}$ such that any $u_1$ elements of the set are independent, for $n^*\leq n_B$. We have that
for $1\leq v\leq n^*$ and $1\leq i_1<i_2\ldots <i_v\leq n_B$, the set $\cap_{j=1}^v \overline E_{i_j}$
contains $f(v)$ elements with \begin{eqnarray}\label{eq:f(v)}
                 f(v)=
                  \begin{cases}
                   (s-1)^{v}s^{u_1-v}, &  1\leq v\leq u_1, \\
                   m^*,            &  u_1+1\leq v\leq n^*,
                   \end{cases}
\end{eqnarray}
where $m^*=s^{u_1}[1-{v\choose 1}s^{-1} + \cdots +(-1)^{u_1}{v\choose u_1}s^{-u_1}] + \sum_{i=u_1+1}^v(-1)^i{v\choose i}$.
\end{proposition}

The value of $n^*$  in Proposition \ref{prop:s-level} will be studied in Section 3.3. Example \ref{ex:1} provides an illustration of the ${\bf b}_i$'s, $\mathcal{A}_i$'s, $\overline E_i$'s and Proposition
\ref{prop:s-level}.

\begin{example}\label{ex:1}
Consider $s=3$, $u=4$ and $u_1=3$. By definition, we have ${\bf e}_1=(1,0,0,0)^T$, ${\bf e}_2=(0,1,0,0)^T$ and ${\bf e}_3=(0,0,1,0)^T$,  $\mathcal{A}=\{(x_1, x_2, x_3, x_4)^T \ | \ x_1=1, x_2, x_3\in \{1, 2\}, x_4\in \{0,1,2\}\}$,
 $n_B=(3-1)^{3-1}=4$, ${\bf b}_1=(1,1,1,0)^T$, ${\bf b}_2=(1, 1, 2,0)^T$, ${\bf b}_3=(1, 2, 1,0)^T$, and ${\bf b}_4=(1,2,2,0)^T$. The disjoint groups $\mathcal{A}_1, \ldots, \mathcal{A}_4$ are displayed in Table \ref{tb0}. Note that any three of $\{{\bf b}_1, {\bf b}_2, {\bf b}_3, {\bf b}_4\}$
are independent. According to (\ref{eq:f(v)}), we have $f(1)=18$, $f(2)=12$,
$f(3)=8$ and $f(4)=6$. That is, each of $\overline E_i$'s has $18$ vectors, as shown in Table \ref{tb1};
the intersection of any two of $\overline E_i$'s has
$12$ vectors, the intersection of any three of $\overline E_i$'s has $8$
vectors, and the intersection of four of them has $6$ vectors.

\begin{table}[h]
{\tabcolsep=6pt
\renewcommand{\arraystretch}{1}
 \begin{center}
 \caption{Partition of $\mathcal{A}$ in Example~\ref{ex:1}\label{tb0}}
 \scalebox{0.8}{
\begin{tabular}{ccc c  ccc c   ccc c  ccc} \hline
 \multicolumn{3}{c}{$\mathcal{A}_1$ } & \multicolumn{1}{c}{} &\multicolumn{3}{c}{$\mathcal{A}_2$} & \multicolumn{1}{c}{}&\multicolumn{3}{c}{$\mathcal{A}_3$ } & \multicolumn{1}{c}{}& \multicolumn{3}{c}{$\mathcal{A}_4$ } \\\hline
 1 &  1   &  1&& 1 & 1& 1&& 1 & 1& 1&& 1 & 1 & 1 \\[-5pt]
 1 &  1   &  1&& 1 & 1& 1&& 2 & 2& 2&& 2 & 2 & 2 \\[-5pt]
 1 &  1   &  1&& 2 & 2& 2&& 1 & 1& 1&& 2 & 2 & 2 \\[-5pt]
 0 &  1   &  2&& 0 & 1& 2&& 0 & 1& 2&& 0 & 1 & 2 \\\hline
\end{tabular}}
\end{center}}
 \end{table}
 \begin{table}[h]
{\tabcolsep=10pt
\renewcommand{\arraystretch}{1}
 \begin{center}
 \caption{Vectors of $\overline E_i$'s in Example~\ref{ex:1}\label{tb1}}
 \scalebox{0.8}{
  \begin{tabular}{   ccc ccc ccc  ccc ccc ccc} 
              \multicolumn{18}{c}{$\overline E_1$}\\ \hline
                     0&0 &0 &1 &1 &1 &1 &1 &1 & 0 &0 &0 &2 &2 &2 &2 &2 &2\\[-5pt]
                     0&1 &1 &0 &0 &1 &2 &1 &2 & 0 &2 &2 &0 &0 &2 &1 &2 &1\\[-5pt]
                     1&0 &1 &0 &1 &0 &2 &2 &1 & 2 &0 &2 &0 &2 &0 &1 &1 &2\\[-5pt]
                     0&0 &0 &0 &0 &0 &0 &0 &0 & 0 &0 &0 &0 &0 &0 &0 &0 &0\\\hline
                   \multicolumn{18}{c}{$\overline E_2$}\\ \hline
                    0 &0  &0  &1  &1  &1  &1  &1  &1 &0 &0 &0 &2 &2 &2 &2 &2 &2\\[-5pt]
                    0 &1  &1  &0  &0  &1  &1  &2  &2 &0 &2 &2 &0 &0 &2 &2 &1 &1\\[-5pt]
                    1 &0  &2  &0  &2  &0  &1  &2  &1 &2 &0 &1 &0 &1 &0 &2 &1 &2\\[-5pt]
                    0 &0  &0  &0  &0  &0  &0  &0  &0 &0 &0 &0 &0 &0 &0 &0 &0 &0\\ \hline
                  \multicolumn{18}{c}{$\overline E_3$}\\ \hline
                    0 & 0 & 0& 1& 1 &1  &1  &1  &1 &0&0&0&2&2&2&2&2&2\\[-5pt]
                    0 & 1 & 1& 0& 0 &2  &1  &2  &1 &0&2&2&0&0&1&2&1&2\\[-5pt]
                    1 & 0 & 2& 0& 1 &0  &1  &2  &2 &2&0&1&0&2&0&2&1&1\\[-5pt]
                    0 & 0 & 0& 0& 0 &0  &0  &0  &0 &0&0&0&0&0&0&0&0&0\\ \hline
                  \multicolumn{18}{c}{$\overline E_4$}\\ \hline
                   0  & 0 & 0 & 1&1 &1  &1 &1 &1 &0&0&0&2&2&2&2&2&2\\[-5pt]
                   0  & 1 & 1 & 0&0 &2  &1 &1 &2 &0&2&2&0&0&1&2&2&1\\[-5pt]
                   1  & 0 & 1 & 0&2 &0  &1 &2 &1 &2&0&2&0&1&0&2&1&2\\[-5pt]
                   0  & 0 & 0 & 0&0 &0  &0 &0 &0 &0&0&0&0&0&0&0&0&0\\ \hline
\end{tabular}}
\end{center}}

\end{table}
\end{example}

Next, we show how to use ${\bf b}_i$, $\mathcal{A}_i$ and $\overline E_i$ ($i=1,\ldots,n_B$) to construct marginally coupled designs. To do so, we define $E_v^*$, $\mathcal{A}_v^*$ and $g(v)$ as follows. To define  $E_v^*$,
given $s$, $u$ and $u_1$,  find a set of
$\{{\bf b}_{i_1}, \ldots, {\bf b}_{i_{n^*}}\}$, by calculation or computer search, such that any $u_1$ elements in the set are independent;  for $1\leq v\leq n^*$, obtain
$\cap_{j=1}^v\overline E_{i_j}$ which
has $f(v)$ elements as shown in Proposition \ref{prop:s-level}. Define $E_v^*$ to be the subset of
$\cap_{j=1}^v\overline E_{i_j}$ in which the first nonzero entry of each element is equal to 1.
The value $g(v)=f(v)/(s-1)$ is the number of elements of $E_v^*$. Define $\mathcal{A}_v^*=\cup_{j=1}^v\mathcal{A}_{i_j}$.

\begin{theorem}\label{theorem:subspace-construction}
 For $E_v^*$, $\mathcal{A}_v^*$ and $g(v)$ defined above, if in the general construction, we
\begin{itemize}
  \item[(i)] choose ${\bf z}_i \in   E_v^*$
             and ${\bf x}_j \in \mathcal{A}_v^*$, $i=1,\ldots,g(v)$ and $j=1,\ldots,vs^{u-u_1}$, an $\text{MCD}(D_1, D_2)$ with $\ D_1=\OA(s^u, g(v), s, 2), D_2=\LHD(s^u, vs^{u-u_1})$
             can be obtained, or
  \item[(ii)] choose ${\bf z}_i \in \mathcal{A}_v^*$
              and ${\bf x}_j \in E_v^*$, $i =1,\ldots, vs^{u-u_1}$ and $j =1, \ldots, g(v)$, an $\text{MCD}(D_1, D_2)$ with $\ D_1=\OA(s^u, vs^{u-u_1}, s, 2), D_2=\LHD(s^u, g(v))$ can be obtained,
\end{itemize}
where both $D_2$'s are non-cascading Latin hypercubes.
\end{theorem}

For ease of the presentation, the method in Theorem
\ref{theorem:subspace-construction} is called {\em subspace construction}.
Example~\ref{ex:2} provides a detailed illustration of obtaining marginally coupled designs via
the subspace construction using the  $\mathcal{A}_i$'s and $\overline E_{i}$'s
in Example \ref{ex:1}.

\begin{example}\label{ex:2}
(Continuation of  Example \ref{ex:1})
Table \ref{tb:thm3andthm4} presents  ${\bf \MCD}(D_1, D_2)$'s
obtained according to the subspace construction method by
choosing $v=1,2,3$ or $4$.
As an illustration, we provide the detailed steps of applying item $(i)$ of Theorem \ref{theorem:subspace-construction} for $v=3$.
Consider the sets $\cap_{j=1}^3 \overline E_j$ and $\cup_{j=1}^3 \mathcal{A}_{j}$.
In {Step 1}, $f(3)=8$, hence   $g(3)=4$. The four elements in $\cap_{j=1}^3 \overline E_j$ with the first nonzero entry being 1
are ${\bf z}_1=(0,0,1,0)^T, {\bf z}_2=(0,1,0,0)^T, {\bf z}_3=(1,0,0,0)^T$, and ${\bf z}_4=(1,2,2,0)^T$;
take $({\bf z}_1, {\bf z}_2, {\bf z}_3, {\bf z}_4)$ as a generator matrix to obtain
$D_1=({\bf a}_1, {\bf a}_2, {\bf a}_3, {\bf a}_4)$, an $\OA(81,4,3,2)$.
In {Step 2}, the $3\cdot 3^{4-3}=9$ elements in
$\cup_{j=1}^3 \mathcal{A}_{j}=\{{\bf x}_1, {\bf x}_2, \ldots, {\bf x}_9\}$ are shown in Table \ref{tb0}.
For each ${\bf x}_i$, let $G({\bf x}_i)$ consist of three independent columns of $O({\bf x}_i)$,
and take  $G({\bf x}_i)$ as a generator matrix to obtain the matrix $A_i$, an $\OA(81, 3, 3, 3)$;
let ${\bf d}_i=A_i\cdot(3^2, 3, 1)^T$, and further let $\tilde D_2=({\bf d}_1, \ldots, {\bf d}_9)$,  an $\OA(81, 9, 27, 1)$.
In {Step 3}, construct $D_2$, an $\LHD(81, 9)$, from $\tilde D_2$ by the {\it level-replacement based Latin hypercube} approach. The above three-step procedure results in an $\MCD(D_1,D_2)$,
  which is listed  in Table \ref{tb:thm3andthm4} marked by $\#$, and in the middle of  Table \ref{tb-designs} marked by $\diamondsuit$.

\begin{table}[h]
{\tabcolsep=6pt
\renewcommand{\arraystretch}{1}
 \begin{center}
 \caption{$MCD(D_1,D_2)$'s with $s=3$, $u=4$ and $u_1=3$ in Example~\ref{ex:2}}\label{tb:thm3andthm4}
 \scalebox{0.8}{
\begin{tabular}{c cc cc} \hline
   &  \multicolumn{2}{c}{By item $(i)$}   & \multicolumn{2}{c}{By item $(ii)$} \\
   \hline
            $v$              & ${D}_1$ &  ${D}_2$           & ${D}_1$ &  ${D}_2$    \\\hline
 $1$   &  $ \OA(3^4, 9, 3, 2)$& $ \LHD(3^4, 3)$  & $ \OA(3^4, 3, 3, 2)$&$ \LHD(3^4, 9)$ \\[-5pt]
 $2$   &  $ \OA(3^4, 6, 3, 2)$& $ \LHD(3^4, 6)$  & $ \OA(3^4, 6, 3, 2)$&$ \LHD(3^4, 6)$ \\[-5pt]
 \hspace{-3mm}$^{\#}3$ &  $ \OA(3^4, 4, 3, 2)$& $ \LHD(3^4, 9)$  & $ \OA(3^4, 9, 3, 2)$&$ \LHD(3^4, 4)$ \\[-5pt]
 $4$   &  $ \OA(3^4, 3, 3, 2)$& $ \LHD(3^4, 12)$ & $ \OA(3^4, 12, 3, 2)$&$ \LHD(3^4, 3)$ \\ \hline
\end{tabular}}
\end{center}}
 \end{table}
\end{example}

\subsection{The maximum value of $n^*$}

Both Proposition \ref{prop:s-level} and Theorem \ref{theorem:subspace-construction} require a set of vectors $\{{\bf b}_{i_1}, \ldots, {\bf b}_{i_{n^*}}\}$ in which
any $u_1$ elements are independent.
The value of $n^*$ directly determines the number of columns in $D_1$ or $D_2$.
Of theoretical interest is the maximum value of $n^*$ that can be achieved,
and the bound of such  a value if not obtained explicitly. We provide the maximum value of $n^*$ for the three cases: (1) $s=2$ with $u_1\geq 2$, (2) $s>2$ with $u_1=1$, and
(3) $s>2$ with $u_2=2$. For other values of $s$, $u$, and $u_1$,
we provide bounds of the maximum value of $n^*$.

\noindent{\bf Case 1:} $s=2$, $u_1\geq 2$

For $s=2$, and  $1\leq u_1<u$, we have $n_B=(s-1)^{u_1-1}=1$ and thus $n^*=1$. The only choice
for ${\bf b}_i$'s, $\mathcal{A}_i$'s and $\overline E_i$'s is ${\bf b}_1=(1,\ldots,1, 0, \ldots, 0)$,
$\mathcal{A}=\mathcal{A}_1=\{(1, \ldots, 1, x_{u_1+1}, \ldots, x_{u})\mid x_i \in \{0, 1\}\}$,
and $\overline E_1$ contains all the combinations of $\lambda_1{\bf e}_1 +{\cdots} + \lambda_{u_1}{\bf e}_{u_1}$ that are
not orthogonal to column vectors of $\mathcal{A}_1$. Note that $\overline E_1$   consists of all combinations with odd numbers
of $\{{\bf e}_1, \ldots, {\bf e}_{u_1}\}$. Therefore, $\overline E_1$ has $2^{u_1-1}$ elements. In addition, $v=1$, $f(1)=g(1)=2^{u_1-1}$ and $k=1\cdot2^{u-u_1}$.

\noindent{\bf Case 2:} $s\geq 3$, $u_1=1$

As $u_1=1$, we have $n_B = (s-1)^{u_1-1} = 1$ and $n^*=1$.
It is clear that $\mathcal{A}=\mathcal{A}_1$, $\overline E_1=\{\alpha{\bf e}_1\mid \alpha\in GF(s)\setminus\{0\}\}$,
$v=1$, $f(1)=s-1$, $g(1)=1$ and $k=1\cdot s^{u-1}$.

\vspace{3mm}
\noindent{\bf Case 3:} $s\geq 3$, $u_1=2$

We have $n_B=(s-1)^{u_1-1}=s-1$. The first $u_1$ entries of vectors of $\mathcal{A}$ have
$s-1$ choices as $(1,\alpha_1)^T, (1, \alpha_2)^T, \ldots, (1, \alpha_{s-1})^T$ for $\alpha_i \in GF(s)$,
hence ${\bf b}_i=(1, \alpha_i, 0, \ldots, 0)^T$. As any two vectors of
$\{{\bf b}_1, {\bf b}_2, \ldots, {\bf b}_{s-1}\}$ are independent, the maximum value
of $n^*$ is $s-1$. The values of $f(v)$ at $v=1, 2$, and $2< v\leq s-1$ are $s(s-1)$, $(s-1)^2$ and $(s-1)(s-v+1)$
according to (\ref{eq:f(v)}), respectively.  The values of $g(v)$ at $v=1, 2$, and $2< v\leq s-1$ are $s$, $s-1$ and $s-v+1$, respectively.

Table \ref{n-star} summarizes the maximum values of $n^*$ under cases 1 to 3, where
the marginally coupled designs are obtained as in Theorem \ref{theorem:subspace-construction}.
For $s=2$, $D_1$ is an orthogonal array of strength three follows by Corollary 2
of Deng, Hung and Lin (2015). For $s, u_1>2$, Proposition \ref{prop:upperbound}
presents a bound for the maximum value of $n^*$.

\begin{table}
{\tabcolsep=6pt
\renewcommand{\arraystretch}{0.8}
 \begin{center}
 \caption{Maximum values of $n^*$ and $MCD(D_1,D_2)$'s for $s=2$ or $u_1\leq 2$\label{n-star}}
 \scalebox{0.7}{
\begin{tabular}{ccccc ll} \hline
 $s$                                 & $u_1$                               &  maximum value of $n^*$ &$v$ & $g(v)$  &\multicolumn{1}{c}{${D}_1$}         & \multicolumn{1}{c}{${D}_2$}\\\hline
  \multirow{2}{2cm}{$s=2$}           & \multirow{2}{2cm}{$2\leq u_1\leq u$}&   \multirow{2}{2cm}{$1$}& 1  &$2^{u_1-1}$& $\OA(2^u, 2^{u_1-1}, 2, 3)$        & $\LHD(2^u, 2^{u-u_1})$ \\
                                     &                                     &                         & 1  &$2^{u_1-1}$& $\OA(2^u, 2^{u-u_1}, 2, 3)$        & $\LHD(2^u, 2^{u_1-1})$ \\\hline
  \multirow{2}{2cm}{$s\geq 3$}       & \multirow{2}{2cm}{$1$}              & \multirow{2}{2cm}{$1$}  & 1  &    1      & $\OA(s^u, 1, s, 2)$        & $\LHD(s^u, s^{u-1})$           \\
                                     &                                     &                         & 1  &    1      & $\OA(s^u, s^{u-1}, s, 2)$  & $\LHD(s^u, 1)$           \\ \hline
  \multirow{6}{2cm}{$s\geq 3$}       & \multirow{6}{2cm}{$2$}              & \multirow{6}{2cm}{$s-1$}& 1  &   $s$     & $\OA(s^u, s, s, 2)$        & $\LHD(s^u, s^{u-2})$           \\
                                     &                                     &                         & 1  &   $s$        & $\OA(s^u, s^{u-2}, s, 2)$  & $\LHD(s^u, s)$           \\
                                     &                                     &                         & 2  &   $s-1$      & $\OA(s^u, s-1, s, 2)$      & $\LHD(s^u, 2s^{u-2})$         \\
                                     &                                     &                         & 2  &   $s-1$      & $\OA(s^u, 2s^{u-2}, s, 2)$ & $\LHD(s^u, s-1)$          \\
                                     &                                     &                         &$2<v\leq s-1$& $s-v+1$& $\OA(s^u, s-v+1, s, 2)$    & $\LHD(s^u, vs^{u-2})$        \\
                                     &                                     &                      &$2<v\leq s-1$& $s-v+1$& $\OA(s^u, vs^{u-2}, s, 2)$ & $\LHD(s^u, s-v+1)$        \\ \hline
  \end{tabular}}
\end{center}}
 \end{table}

\begin{proposition}\label{prop:upperbound}
Given positive integers $s, u>2$, and $2<u_1\leq u$, suppose any $u_1$ vectors of $\{{\bf b}_1, \ldots, {\bf b}_{n^*}\}$ are independent.  We have
\begin{eqnarray}\label{eq:$n^*$}
                 \max n^*\leq
                  \begin{cases}
                   u_1+1,      &  s\leq u_1, \\
                   s+u_1-2,    &  s>u_1\geq 3 \ \mbox{and} \ s \ \mbox{is odd},\\
                   s+u_1-1,    &  \mbox{in all other cases}.
                   \end{cases}
 \end{eqnarray}
\end{proposition}


\noindent{\bf Remark 1.}   According to the proof of Proposition \ref{prop:upperbound}, the maximum value of $n^*$ is not greater than the maximum value of $m$ in an $OA(s^{u_1}, m, s, u_1)$. It shall be noted that, however, it is possible to give an upper bound tighter than that given by Proposition \ref{prop:upperbound},  for example, for $u_1=2$, the maximum value of $n^*$ is $s-1$, but the maximum value of $m$ in an $\OA(s^2, m, s, 2)$ is {\bf $s+1$}.

\section{Tables for Three-level Qualitative Factors}

This section tabulates the marginally coupled designs with three-level qualitative factors obtained by the proposed methods for practical use.
Tables \ref{tb-simpleconstructiondesigns} and \ref{tb-designs} present the designs constructed in Theorems \ref{thm-simple-construction} and \ref{theorem:subspace-construction}, respectively, where $\overline u_1=u-u_1$,
and the symbol $*$ indicates the case of $v=n^*$.

\begin{table}
{\tabcolsep=10pt
\renewcommand{\arraystretch}{0.8}
 \begin{center}
 \caption{$MCD(D_1, D_2)$s with $3^u$ runs by Theorem \ref{thm-simple-construction}, $u=2,3,4,5$\label{tb-simpleconstructiondesigns}}
 \scalebox{0.8}{
  \begin{tabular}{c c c ll ll} \hline
  \multirow{2}{0.5cm}{$u$} & \multirow{2}{0.5cm}{$u_1$} &\multirow{2}{0.5cm}{$n_A$} & \multicolumn{2}{c}{By item $(i)$ }                              &   \multicolumn{2}{c}{By item $(ii)$}\\\cline{4-7}
                           &                          &                         & \multicolumn{1}{c}{${D}_1$}  & \multicolumn{1}{c}{${D}_2$}  & \multicolumn{1}{c}{${D}_1$} & \multicolumn{1}{c}{${D}_2$} \\ \hline
   2  &   1   &   3                                                             & $\OA(3^2, 1, 3, 1)$               & $\LHD(3^2, 3)$                & $\OA(3^2, 3, 3, 2)$             &  $\LHD(3^2, 1)$\\
   2  &   2   &   2                                                             & $\OA(3^2, 2, 3, 2)$               & $\LHD(3^2, 2)$                & $\OA(3^2, 2, 3, 2)$             &  $\LHD(3^2, 2)$ \\
   3  &   1   &   9                                                             & $\OA(3^3, 1, 3, 1)$               & $\LHD(3^3, 9)$                & $\OA(3^3, 9, 3, 2)$             &  $\LHD(3^3, 1)$\\
   3  &   2   &   6                                                             & $\OA(3^3, 2, 3, 2)$               & $\LHD(3^3, 6)$                & $\OA(3^3, 6, 3, 2)$             &  $\LHD(3^3, 2)$\\
   3  &   3   &   4                                                             & $\OA(3^3, 3, 3, 3)$               & $\LHD(3^3, 4)$                & $\OA(3^3, 4, 3, 2)$             &  $\LHD(3^3, 3)$\\
   4  &   1   &   27                                                            & $\OA(3^4, 1, 3, 1)$               & $\LHD(3^4, 27)$               & $\OA(3^4, 27, 3, 2)$            &  $\LHD(3^4, 1)$\\
   4  &   2   &   18                                                            & $\OA(3^4, 2, 3, 2)$               & $\LHD(3^4, 18)$               & $\OA(3^4, 18, 3, 2)$            &  $\LHD(3^4, 2)$\\
   4  &   3   &   12                                                            & $\OA(3^4, 3, 3, 3)$               & $\LHD(3^4, 12)$               & $\OA(3^4, 12, 3, 2)$            &  $\LHD(3^4, 3)$\\
   4  &   4   &    8                                                            & $\OA(3^4, 4, 3, 4)$               & $\LHD(3^4, 8)$                & $\OA(3^4, 8, 3, 2)$             &  $\LHD(3^4, 4)$\\
   5  &   1   &   81                                                            & $\OA(3^5, 1, 3, 1)$               & $\LHD(3^5, 81)$               & $\OA(3^5, 81, 3, 2)$            &  $\LHD(3^5, 1)$\\
   5  &   2   &   54                                                            & $\OA(3^5, 2, 3, 2)$               & $\LHD(3^5, 54)$               & $\OA(3^5, 54, 3, 2)$            &  $\LHD(3^5, 2)$\\
   5  &   3   &   36                                                            & $\OA(3^5, 3, 3, 3)$               & $\LHD(3^5, 36)$               & $\OA(3^5, 36, 3, 2)$            &  $\LHD(3^5, 3)$\\
   5  &   4   &   24                                                            & $\OA(3^5, 4, 3, 4)$               & $\LHD(3^5, 24)$               & $\OA(3^5, 24, 3, 2)$            &  $\LHD(3^5, 4)$ \\
   5  &   5   &   16                                                            & $\OA(3^5, 5, 3, 5)$               & $\LHD(3^5, 16)$               & $\OA(3^5, 16, 3, 2)$            &  $\LHD(3^5, 5)$\\\hline
 \end{tabular}}
\end{center}}
\end{table}


 \begin{table}
{\tabcolsep=10pt
\renewcommand{\arraystretch}{0.8}
 \begin{center}
 \caption{$MCD(D_1, D_2)$s with $3^u$ runs by Theorem \ref{theorem:subspace-construction}, $u=2,3,4,5$\label{tb-designs}}
 \scalebox{0.7}{
  \begin{tabular}{c c l c c c ll ll} \hline
  \multirow{2}{0.5cm}{$u$}    &  \multirow{2}{0.5cm}{$u_1$} &\multirow{2}{0.5cm}{$v$} &  \multirow{2}{0.5cm}{$g(v)$}  & \multirow{2}{0.5cm}{$\overline u_1$}& \multirow{2}{0.5cm}{$k$}                                                                                                                      & \multicolumn{2}{c}{By item $(i)$} & \multicolumn{2}{c}{By item $(ii)$}  \\\cline{7-10}
                              &&&&&    &\multicolumn{1}{c}{${D}_1$}   & \multicolumn{1}{c}{${D}_2$} & \multicolumn{1}{c}{${D}_1$} & \multicolumn{1}{c}{${D}_2$}\\\hline
  2    &   1   & 1* &   1    &  1      & 3  & $\mbox{OA}(3^2, 1, 3, 2)$ & $\mbox{LHD}(3^2, 3)$ &$\mbox{OA}(3^2, 3, 3, 2)$  & $\mbox{LHD}(3^2, 1)$\\
  2    &   2   & 1  &   3    &  0      & 1  & $\mbox{OA}(3^2, 3, 3, 2)$ & $\mbox{LHD}(3^2, 1)$ &$\mbox{OA}(3^2, 1, 3, 2)$  & $\mbox{LHD}(3^2, 3)$\\
  2    &   2   & 2* &   2    &  0      & 2  & $\mbox{OA}(3^2, 2, 3, 2)$ & $\mbox{LHD}(3^2, 2)$ & $\mbox{OA}(3^2, 2, 3, 2)$ & $\mbox{LHD}(3^2, 2)$\\
  3    &   1   & 1* &   1    &  2      & 9  & $\mbox{OA}(3^3, 1, 3, 2)$ & $\mbox{LHD}(3^3, 9)$ & $\mbox{OA}(3^3, 9, 3, 2)$ & $\mbox{LHD}(3^3, 1)$\\
  3    &   2   & 1  &   3    &  1      & 3  & $\mbox{OA}(3^3, 3, 3, 2)$ & $\mbox{LHD}(3^3, 3)$ & $\mbox{OA}(3^3, 3, 3, 2)$ & $\mbox{LHD}(3^3, 3)$\\
  3    &   2   & 2* &   2    &  1      & 6  & $\mbox{OA}(3^3, 2, 3, 2)$ & $\mbox{LHD}(3^3, 6)$ & $\mbox{OA}(3^3, 6, 3, 2)$ & $\mbox{LHD}(3^3, 2)$\\
  3    &   3   & 1  &   9    &  0      & 1  & $\mbox{OA}(3^3, 9, 3, 2)$ & $\mbox{LHD}(3^3, 1)$ & $\mbox{OA}(3^3, 1, 3, 2)$ & $\mbox{LHD}(3^3, 9)$\\
  3    &   3   & 2  &   6    &  0      & 2  & $\mbox{OA}(3^3, 6, 3, 2)$ & $\mbox{LHD}(3^3, 2)$ & $\mbox{OA}(3^3, 2, 3, 2)$ & $\mbox{LHD}(3^3, 6)$\\
  3    &   3   & 3  &   4    &  0      & 3  & $\mbox{OA}(3^3, 4, 3, 2)$ & $\mbox{LHD}(3^3, 3)$ & $\mbox{OA}(3^3, 3, 3, 2)$ & $\mbox{LHD}(3^3, 4)$\\
  3    &   3   & 4* &   3    &  0      & 4  & $\mbox{OA}(3^3, 3, 3, 2)$ & $\mbox{LHD}(3^3, 4)$ & $\mbox{OA}(3^3, 4, 3, 2)$ & $\mbox{LHD}(3^3, 3)$\\
  4    &   1   & 1* &   1    &  3      & 27 & $\mbox{OA}(3^4, 1, 3, 2)$ & $\mbox{LHD}(3^4, 27)$& $\mbox{OA}(3^4, 27, 3, 2)$& $\mbox{LHD}(3^4, 1)$\\
  4    &   2   & 1  &   3    &  2      & 9  & $\mbox{OA}(3^4, 3, 3, 2)$ & $\mbox{LHD}(3^4, 9)$ & $\mbox{OA}(3^4, 9, 3, 2)$ & $\mbox{LHD}(3^4, 3)$\\
  4    &   2   & 2* &   2    &  2      & 18 & $\mbox{OA}(3^4, 2, 3, 2)$ & $\mbox{LHD}(3^4, 18)$& $\mbox{OA}(3^4, 18, 3, 2)$& $\mbox{LHD}(3^4, 2)$\\
  4    &   3   & 1  &   9    &  1      & 3  & $\mbox{OA}(3^4, 9, 3, 2)$ & $\mbox{LHD}(3^4, 3)$ & $\mbox{OA}(3^4, 3, 3, 2)$ & $\mbox{LHD}(3^4, 9)$\\
  4    &   3   & 2  &   6    &  1      & 6  & $\mbox{OA}(3^4, 6, 3, 2)$ & $\mbox{LHD}(3^4, 6)$ & $\mbox{OA}(3^4, 6, 3, 2)$ & $\mbox{LHD}(3^4, 6)$\\
  \hspace{-3mm}$^\diamondsuit$4    &   3   & 3  &   4    &  1      & 9  & $\mbox{OA}(3^4, 4, 3, 2)$ & $\mbox{LHD}(3^4, 9)$ & $\mbox{OA}(3^4, 9, 3, 2)$ & $\mbox{LHD}(3^4, 4)$\\
  4    &   3   & 4* &   3    &  1      & 12 & $\mbox{OA}(3^4, 3, 3, 2)$ & $\mbox{LHD}(3^4, 12)$& $\mbox{OA}(3^4, 12, 3, 2)$& $\mbox{LHD}(3^4, 3)$\\
  4    &   4   & 1  &   27   &  0      & 1  & $\mbox{OA}(3^4, 27, 3, 2)$& $\mbox{LHD}(3^4, 1)$ & $\mbox{OA}(3^4, 1, 3, 2)$ & $\mbox{LHD}(3^4, 27)$\\
  4    &   4   & 2  &   18   &  0      & 2  & $\mbox{OA}(3^4, 18, 3, 2)$& $\mbox{LHD}(3^4, 2)$ & $\mbox{OA}(3^4, 2, 3, 2)$ & $\mbox{LHD}(3^4, 18)$\\
  4    &   4   & 3  &   12   &  0      & 3  & $\mbox{OA}(3^4, 12, 3, 2)$& $\mbox{LHD}(3^4, 3)$ & $\mbox{OA}(3^4, 3, 3, 2)$ & $\mbox{LHD}(3^4, 12)$\\
  4    &   4   & 4  &   8    &  0      & 4  & $\mbox{OA}(3^4, 8,  3, 2)$& $\mbox{LHD}(3^4, 4)$ & $\mbox{OA}(3^4, 4,  3, 2)$& $\mbox{LHD}(3^4, 8)$\\
  4    &   4   & 5* &   5    &  0      & 5  & $\mbox{OA}(3^4, 5, 3, 2)$ &$\mbox{LHD}(3^4, 5)$  & $\mbox{OA}(3^4, 5, 3, 2)$ & $\mbox{LHD}(3^4, 5)$\\
  5    &   1   & 1* &   1    &  4      & 81 & $\mbox{OA}(3^5, 1, 3, 2)$& $\mbox{LHD}(3^5, 81)$& $\mbox{OA}(3^5, 81, 3, 2)$& $\mbox{LHD}(3^5, 1)$\\
  5    &   2   & 1  &   3    &  3      & 27 & $\mbox{OA}(3^5, 3, 3, 2)$& $\mbox{LHD}(3^5, 27)$& $\mbox{OA}(3^5, 27, 3, 2)$& $\mbox{LHD}(3^5, 3)$\\
  5    &   2   & 2* &   2    &  3      & 54 & $\mbox{OA}(3^5, 2, 3, 2)$& $\mbox{LHD}(3^5, 54)$& $\mbox{OA}(3^5, 54, 3, 2)$& $\mbox{LHD}(3^5, 2)$\\
  5    &   3   & 1  &   9    &  2      & 9  & $\mbox{OA}(3^5, 9, 3, 2)$& $\mbox{LHD}(3^5, 9)$& $\mbox{OA}(3^5, 9, 3, 2)$& $\mbox{LHD}(3^5, 9)$\\
  5    &   3   & 2  &   6    &  2      & 18 & $\mbox{OA}(3^5, 6, 3, 2)$& $\mbox{LHD}(3^5, 18)$& $\mbox{OA}(3^5, 18, 3, 2)$& $\mbox{LHD}(3^5, 6)$\\
  5    &   3   & 3  &   4    &  2      & 27 & $\mbox{OA}(3^5, 4, 3, 2)$& $\mbox{LHD}(3^5, 27)$& $\mbox{OA}(3^5, 27, 3, 2)$& $\mbox{LHD}(3^5, 4)$\\
  5    &   3   & 4* &   3    &  2      & 36 & $\mbox{OA}(3^5, 3, 3, 2)$& $\mbox{LHD}(3^5, 36)$& $\mbox{OA}(3^5, 36, 3, 2)$& $\mbox{LHD}(3^5, 3)$\\
  5    &   4   & 1  &   27   &  1      & 3  & $\mbox{OA}(3^5, 27, 3, 2)$& $\mbox{LHD}(3^5, 3)$& $\mbox{OA}(3^5, 3, 3, 2)$& $\mbox{LHD}(3^5, 27)$\\
  5    &   4   & 2  &   18   &  1      & 6  & $\mbox{OA}(3^5, 18, 3, 2)$& $\mbox{LHD}(3^5, 6)$& $\mbox{OA}(3^5, 6, 3, 2)$& $\mbox{LHD}(3^5, 18)$\\
  5    &   4   & 3  &   12   &  1      & 9  & $\mbox{OA}(3^5, 12, 3, 2)$& $\mbox{LHD}(3^5, 9)$& $\mbox{OA}(3^5, 9, 3, 2)$& $\mbox{LHD}(3^5, 12)$\\
  5    &   4   & 4  &    8   &  1      & 12 & $\mbox{OA}(3^5, 8, 3, 2)$& $\mbox{LHD}(3^5, 12)$& $\mbox{OA}(3^5, 12, 3, 2)$& $\mbox{LHD}(3^5, 8)$\\
  5    &   4   & 5* &    5   &  1      & 15 & $\mbox{OA}(3^5, 5, 3, 2)$& $\mbox{LHD}(3^5, 15)$& $\mbox{OA}(3^5, 15, 3, 2)$& $\mbox{LHD}(3^5, 5)$\\
  5    &   5   & 1  &   81   &  0      & 1  & $\mbox{OA}(3^5, 81, 3, 2)$& $\mbox{LHD}(3^5, 1)$& $\mbox{OA}(3^5, 1, 3, 2)$& $\mbox{LHD}(3^5, 81)$\\
  5    &   5   & 2  &   54   &  0      & 2  & $\mbox{OA}(3^5, 54, 3, 2)$& $\mbox{LHD}(3^5, 2)$& $\mbox{OA}(3^5, 2, 3, 2)$& $\mbox{LHD}(3^5, 54)$\\
  5    &   5   & 3  &   36   &  0      & 3  & $\mbox{OA}(3^5, 36, 3, 2)$& $\mbox{LHD}(3^5, 3)$& $\mbox{OA}(3^5, 3, 3, 2)$& $\mbox{LHD}(3^5, 36)$\\
  5    &   5   & 4  &   24   &  0      & 4  & $\mbox{OA}(3^5, 24, 3, 2)$& $\mbox{LHD}(3^5, 4)$& $\mbox{OA}(3^5, 4, 3, 2)$& $\mbox{LHD}(3^5, 24)$\\
  5    &   5   & 5  &   16   &  0      & 5  & $\mbox{OA}(3^5, 16, 3, 2)$& $\mbox{LHD}(3^5, 5)$& $\mbox{OA}(3^5, 5, 3, 2)$& $\mbox{LHD}(3^5, 16)$\\
  5    &   5   & 6* &   11   &  0      & 6  & $\mbox{OA}(3^5, 11, 3, 2)$& $\mbox{LHD}(3^5, 6)$& $\mbox{OA}(3^5, 6, 3, 2)$& $\mbox{LHD}(3^5, 11)$\\
  \hline
 \end{tabular}}
\end{center}}
\end{table}


Since the last $u-u_1$ entries of each ${\bf b}_i$ are zeros, to obtain the maximum value
of $n^*$, we only need to
consider the independent relationship between  the vectors with the first $u_1$ entries of ${\bf b}_i$'s.
For $s=3$, $n_B=2^{u_1-1}$ and these vectors can form a $u_1\times 2^{u_1-1}$ matrix,
which is denoted by $B_{u_1}$ in this paper. Columns of $B_{u_1}$  are arranged in an order such that
the $j$th column is determined by the $(i,j)$th entry $B_{u_1}(i,j)$ as follows:
\begin{eqnarray}\nonumber
j-1=\sum_{i=1}^{u_1}2^{u_1-i}(B_{u_1}(i,j)-1).
\end{eqnarray}
Hence the $j$th column is labeled by bold ${\bf j-1}$ in Table \ref{tb-B2}, in which
the matrices of $B_2$ to $B_5$ are presented.
Correspondingly, define $B_{u_1}^*$ to be an $n^*$-column subset of $B_{u_1}$,
such that any $u_1$ columns in it are independent.
The following is a list of the sets $B_2^*$ to $B_5^*$:
$B_2^*$ containing columns $\{{\bf 0, \bf 1}\}$ of $B_2$;
$B_3^*$ containing columns $\{{\bf 0, \bf 1, \bf 2, \bf 3}\}$ of $B_3$;
$B_4^*$ containing columns $\{{\bf 0, \bf 1,\bf 2, \bf 4, \bf 7}\}$ of $B_4$; and
$B_5^*$ containing columns $\{{\bf 0, \bf 1,\bf 2, \bf 4, \bf 9, \bf 14}\}$ of $B_5$,
where $B_2^*$ and $B_3^*$ are obtained by calculation, and $B_4^*$ and $B_5^*$ are
obtained by computer search.
All of their $n^*$'s are maximal, refer to Proposition \ref{prop:upperbound}.
With those $B_{u_1}^*$'s, one can obtain the set of column vectors
$\{{\bf b}_{i_1}, \ldots, {\bf b}_{i_{n^*}}\}$ required by Theorem \ref{theorem:subspace-construction}.

 \begin{table}[htbp]
{\tabcolsep=12pt
\renewcommand{\arraystretch}{0.8}
 \begin{center}
  \caption{Matrices $B_{u_1}$'s for $u_1=2, 3, 4, 5$ and $s=3$\label{tb-B2}}
   \scalebox{0.8}{
   \begin{tabular}{cc c cccc c cccccccc } \hline
    \multicolumn{2}{c}{$B_2$} &\multicolumn{1}{c}{} & \multicolumn{4}{c}{$B_3$} & \multicolumn{1}{c}{}& \multicolumn{8}{c}{$B_4$} \\
    {\bf 0} & {\bf 1}  && {\bf 0} & {\bf 1}& {\bf 2} & {\bf 3} && {\bf 0}& {\bf 1} & {\bf 2}& {\bf 3}&{\bf 4}& {\bf 5}&{\bf 6}&{\bf 7}\\\hline
     1      & 1        &&1 & 1 & 1 & 1                         && 1 & 1 & 1 & 1 &1 & 1 & 1 & 1\\[-5pt]
     1      & 2        &&1 & 1 & 2 & 2                         && 1 & 1 & 1 & 1 &2 & 2 & 2 & 2\\[-5pt]
            &          &&1 & 2 & 1 & 2                         && 1 & 1 & 2 & 2 &1 & 1 & 2 & 2\\[-5pt]
            &          &&  &   &   &                           && 1 & 2 & 1 & 2 &1 & 2 & 1 & 2\\\hline
            \multicolumn{16}{c}{$B_5$} \\
     {\bf 0}& {\bf 1}& {\bf 2}& {\bf 3}& {\bf 4}& {\bf 5}&{\bf 6}& {\bf 7}& {\bf 8}& {\bf 9}& {\bf 10}&{\bf 11}&{\bf 12} &{\bf 13}&{\bf 14} &{\bf 15} \\\hline
     1 & 1 & 1 & 1 &1 & 1 & 1 & 1 & 1 & 1 & 1 & 1 &1 & 1 & 1 & 1\\[-5pt]
     1 & 1 & 1 & 1 &1 & 1 & 1 & 1 & 2 & 2 & 2 & 2 &2 & 2 & 2 & 2\\[-5pt]
     1 & 1 & 1 & 1 &2 & 2 & 2 & 2 & 1 & 1 & 1 & 1 &2 & 2 & 2 & 2\\[-5pt]
     1 & 1 & 2 & 2 &1 & 1 & 2 & 2 & 1 & 1 & 2 & 2 &1 & 1 & 2 & 2\\[-5pt]
     1 & 2 & 1 & 2 &1 & 2 & 1 & 2 & 1 & 2 & 1 & 2 &1 & 2 & 1 & 2\\
 \hline
 \end{tabular}}
\end{center}}
\end{table}

\section{Space-filling Property}

One important issue of marginally coupled designs is the space-filling property of design $D_2$.   To achieve or improve the space-filling property, several approaches have been proposed; see, for example, Dragulic, Santner and Dean (2012), Joseph, Gul and Ba (2015), and Sun and Tang (2017). In our case,  one approach to improve the space-filling property is to use an optimal level replacement
with some optimization criterion when obtaining $D_2$ from $\tilde{D}_2$, as done in Leary, Bhaskar and Keane (2003); another approach is to make $D_2$ possess some guaranteed space-filling property, for example, having
uniform projections on lower dimensions. In this paper, we address this issue
through the latter approach. For $s=2$, the approach uses a concept, anti-mirror vector, defined below.

\begin{defi}
Two column vectors $v_1$ and $v_2$ of the same length with entries from $\{0,1\}$ are said to be anti-mirror vectors if
their sum is equal to the vector of all ones. We use the notation $\overline v_1=v_2$ and  $\overline v_2=v_1$.
\end{defi}
For example, $(1,1,0)^T$ is the anti-mirror vector of $(0,0,1)^T$.
It is clear that $v^T\overline v=0$, and the anti-mirrors of two different vectors are different.


For practical application, given parameters $1\leq u_1, u'_1\leq u$, item $(ii)$
of Theorem \ref{theorem:subspace-construction} can construct an $MCD(D_1, D_2)$ with
$D_1=\OA(2^u, 2^{u-u_1}, 2, 3)$ and $D_2=\LHD(2^u, 2^{u_1-1})$, and item $(i)$ can construct
an $MCD(D_1, D_2)$ with $D_1=\OA(2^u, 2^{u'_1-1}, 2, 3)$ and $D_2=\LHD(2^u, 2^{u-u'_1})$. When
setting $u_1'=u-u_1+1$, the MCD obtained by item $(i)$ has the same set of parameters
as that obtained by item $(ii)$. In this sense, for $s=2$, we only need to consider the subspace construction
by item $(i)$ of Theorem \ref{theorem:subspace-construction}.

To investigate the space-filling property of $D_2$ when $D_1$ is a two-level orthogonal array, we take a closer look at
 Step 2 of the general construction.
Recall that $\mathcal{A}=\mathcal{A}_1$ has $2^{u-u_1}$ vectors, $n_B=1$
and ${\bf b}_1=(1,\ldots,1,0,\ldots,0)^T$ with
the first $u_1$ entries being 1. As in item (ii) of Theorem~\ref{theorem:subspace-construction},  let
$\{{\bf x}_1, \ldots, {\bf x}_{2^{u-u_1}}\}$  be the vectors in $\mathcal{A}_1$, and note that
each ${\bf x}_i$ can be written as
$${\bf x}_i=({\bf 1}_{u_1}^T, {\bf y}_i^T)^T,$$
where ${\bf y}_i\neq {\bf y}_j$ for $i\neq j$.
Let ${\bf x}_0=(1,1,0,\ldots, 0)^T$
be a vector with the first two entries being 1 and the last $u_1-2$ entries being 0;
for $1\leq i\leq 2^{u-u_1}$, define ${\eta}_i=({\bf x}_0^T, \overline {\bf y}_i^T)^T$,
where $\overline {\bf y}_i$ is the anti-mirror vector of ${\bf y}_i$. We have
$\eta_i\in O({\bf x}_i)$ as ${\eta}_i^T{\bf x}_i={\bf x}_0^T{\bf 1}_{u_1} + \overline {\bf y}_i^T{\bf y}_i=0$
For each ${\bf x}_i$,
let $G({\bf x}_i)$ be a generator matrix that consists of $u-1$ independent columns of $O({\bf x}_i)$.    Set  the first column of $G({\bf x}_i)$   to be  $\eta_i$.  Generate $A_i$ based on $G({\bf x}_i)$ and obtain ${\bf d}_i=A_i\cdot(2^{u-2}, \ldots, 2, 1)^T$,
and let $\tilde D_2=({\bf d}_1, \ldots, {\bf d}_{2^{u-u_1}})$. The method is called the {\em anti-mirror arrangement} in this paper.

\begin{proposition}\label{prop:anti-mirror}
When $2\leq u_1<u-1$, the design $\tilde D_2$ obtained by the anti-mirror arrangement is an $OA(2^u, 2^{u-u_1}, 2^{u-1}, 1)$
achieving stratifications on a $2\times 2\times 2$ grid of any three dimensions.
\end{proposition}

For $s\geq 2$,  Proposition \ref{prop:5} provides a result of the space-filling property of $D_2$'s in marginally coupled designs in Theorem~\ref{theorem:subspace-construction}.

\begin{proposition}\label{prop:5}
If the number, $k$, of columns in $D_2$ in Theorem~\ref{theorem:subspace-construction} satisfies $k\leq (s^{u-1}-1)/(s-1)$, a $\tilde {D}_2$ that achieves stratifications on an $s\times s$ grid of any two dimensions can be constructed.
\end{proposition}

\section{Conclusion and {\bf Discussion}}

We have proposed a general method for constructing marginally coupled designs of $s^u$ runs in which
the design for quantitative factors is a non-cascading Latin hypercube, where $s$ is a prime power.
The approach uses the  theory of $(u-1)$-dimensional subspaces in the Galois field $GF(s^u)$.
The newly constructed marginally coupled designs with three-level qualitative factors are tabulated.  For other prime numbers of levels, marginally coupled designs can be obtained similarly.
  In addition, we discuss two cases for which guaranteed space-filling property can be obtained.

   The  results for the subspace construction in this article extend those in He et al. (2017) for two-level
qualitative factors to any $s$-level qualitative factors. The {\em Construction 2} of He, Lin and Sun (2017)
is also a special case of the general construction in this article. The reason is as follows.
There are $s+1$ matrices of size $s^u\times (s^{u-1}-1)/(s-1)$, denoted by $C_1, \ldots, C_{s+1}$, each of which contains $s$ replications
of the linear saturated orthogonal array $\OA(s^{u-1}, (s^{u-1}-1)/(s-1), s, 2)$.
According to their construction procedure, the matrix $C_i$ is corresponding to the $(u-1)$-dimensional subspace generated by
$\{{\bf e}_1, \ldots, {\bf e}_{u-2}, {\bf e}_{u-1}+\alpha_{i-1}{\bf e}_u\}$ for $1\leq i\leq s$,
and $C_{s+1}$ is corresponding to the $(u-1)$-dimensional subspace generated by
$\{{\bf e}_1, \ldots, {\bf e}_{u-2}, {\bf e}_u\}$.
They are respectively identical to the $(u-1)$-dimensional subspaces $O({\bf x}_1), \ldots, O({\bf x}_{s+1})$,
where ${\bf x}_1={\bf e}_u$, ${\bf x}_i={\bf e}_{u-1}-\alpha_{i-1}^{-1}{\bf e}_u$ for $2\leq i\leq s$, and
${\bf x}_{s+1}={\bf e}_{u-1}$. Therefore, in the general construction, by choosing such ${\bf x}_1, \ldots, {\bf x}_k$,
for $1\leq k < s+1$, and choosing ${\bf z}_1, \ldots, {\bf z}_m$ from the set of
$\cup_{j=k+1}^{s+1}O({\bf x}_j)\setminus (\cup_{i=1}^k O({\bf x}_i))$,
one can obtain the marginally coupled design provided by {\em Construction 2} of He, Lin and Sun (2017).

  For practitioners, three related issues need further investigations. One is that, the low-dimensional projection space-filling property of the quantitative factors for each level of a qualitative factor; the second one is
to improve the space-filling property of the quantitative factors in 3 to 4 dimensions, when the two-dimensional
uniform projections are already obtained; and the last one is to construct designs with good coverage if
perfect space-filling property under some criterion  is not expected. We hope to study them and report our results in future.

\section*{Appendix}
\noindent {\it Proof of Lemma \ref{lemma:general-union-O(e)}}

\begin{proof}
For $1\leq i\leq u_1$ and any vector ${\bf x}=(x_{1}, \ldots, x_u)^T \in S_u \setminus O({\bf e}_i)$, we have ${\bf x}^T{\bf e}_i\neq 0$,
that means $x_i\neq 0$. Thus, for any ${\bf x}\in \mathcal{A}$, we have $x_1=1$, $x_i\in GF(s)\setminus \{0\}$
for $i=2,\ldots, u_1$, and $x_j\in GF(s)$ for $j=u_1+1, \ldots, u$. So,
the conclusion follows.
\end{proof}

\noindent {\it Proof of Theorem \ref{thm-simple-construction}}

\begin{proof}
As every ${\bf z}_i$ is not in any of $O({\bf x}_j)$,  every ${\bf x}_j$ is not in any of $O({\bf z}_i)$.
The conclusion follows by the definition of $\mathcal{A}$,
Lemma \ref{lemma:basic-idea}, and Lemma \ref{lem:D1-D2-condition}.
Because in both items $(i)$ and $(ii)$, $O({\bf x}_i)\neq O({\bf x}_j)$ when $i\neq j$,
${\bf d}_i$ cannot be transformed to ${\bf d}_j$ by level permutations. Thus
$D_2$'s are non-cascading Latin hypercubes.
\end{proof}

\noindent{\it Proof of Proposition \ref{thm-combination is impossible}}

\begin{proof}
Suppose ${\bf z}=\sum_{i=1}^{u_1}\lambda_i {\bf e}_i$  has $l$ nonzero coefficients $\lambda_{i_1}, \ldots, \lambda_{i_l}$,
where $1\leq i_j\leq u_1$ and $ 2\leq l\leq u_1$. Denote by $\lambda^*=\sum_{j=1}^{l-1}\lambda_{i_j}$, and let
${\bf x}=(x_1, \ldots, x_u)^T$. If $\lambda^*$ is nonzero, take $x_{i_l}=-\lambda_{i_l}^{-1}\lambda^*$ and
all the other $x_i$'s equal 1, then ${\bf x}\in\mathcal{A}$ since the first $u_1$ entries of ${\bf x}$ are nonzero.
More specifically, the first entry of ${\bf x}$ is 1, and
$${\bf z}^T{\bf x}=\sum_{i=1}^{u_1}\lambda_ix_i=\sum_{j=1}^{l}\lambda_{{i_j}}x_{i_j}=\sum_{j=1}^{l-1}\lambda_{i_j}\cdot 1
+ \lambda_{i_l}\cdot x_{i_l}=\lambda^* - \lambda_{{i_l}}\cdot\lambda_{{i_l}}^{-1}\lambda^*=0,$$
where the first equality holds because the last $u-u_1$ entries of ${\bf z}$ are zeros.
Otherwise, if $\lambda^*=0$, we must have $l-1\geq 2$, and one can take $x_{i_{l-1}}=\alpha_2$,
$x_{{i_l}}=-\lambda_{{i_l}}^{-1}\lambda_{i_{l-1}}(\alpha_2-1)$, and all other $x_i$'s equal 1.
Note for $s>2$, we have $\alpha_2\neq 1$, hence $x_{{i_l}}\neq 0$ and $\bf x\in \mathcal{A}$ again.
In addition,
$${\bf z}^T{\bf x}=\sum_{i=1}^{u_1}\lambda_ix_i=\sum_{j=1}^{l}\lambda_{{i_j}}x_{i_j}
=\sum_{j=1}^{l-1}\lambda_{{i_j}}\cdot 1 + \lambda_{i_{l-1}}\cdot(\alpha_2 - 1)- \lambda_{{i_l}}\cdot\lambda^{-1}_{{i_l}}\lambda_{i_{l-1}}(\alpha_2-1)=0.$$
So, there always exists an $\bf x\in \mathcal{A}$, such that ${\bf z}\in O({\bf x})$.
\end{proof}

\vspace{4mm}
\noindent{\it Proof of Proposition \ref{prop:s-level}}

\begin{proof}
First, consider $v=1$. As $(\sum_{j=1}^{u_1}\lambda_je_j)^T{\bf b}_1= 0$, we have
\begin{equation}\nonumber
\lambda_1b_{11} + \lambda_2b_{12} + \cdots + \lambda_{u_1}b_{1u_1}=0.
\end{equation}
There are $s^{u_1-1}$ solutions for such an equation, hence there are $s^{u_1}-s^{u_1-1}=(s-1)s^{u_1-1}$
combinations in $\overline E_1$.

For $v=2$, as $(\sum_{j=1}^{u_1}\lambda_je_j)^T{\bf b}_i= 0$ for $i=1,2$, then
\begin{eqnarray}\nonumber
\left\{
   \begin{array}{cc}
    \lambda_1b_{11} + \lambda_2b_{12} + \cdots + \lambda_{u_1}b_{1u_1} &=0, \label{eq:1}\\
    \lambda_1b_{21} + \lambda_2b_{22} + \cdots + \lambda_{u_1}b_{2u_1} &=0, \label{eq:2}
    \end{array}
    \right.
\end{eqnarray}
which has $s^{u_1-2}$ solutions since ${\bf b}_1$ and ${\bf b}_2$ are independent.
However, elements in $\overline E_1\cap \overline E_2$ should not be the solution of
neither of the two equations.
Then, we have
$$\mid \overline E_1\cap \overline E_2 \mid =\mid E\setminus ( E_1\cup E_2)\mid=s^{u_1}- [{2\choose 1}s^{u_1-1} - {2\choose 2}s^{u_1-2}]=(s-1)^2s^{u_1-2}.$$

For $1\leq v\leq u_1$, as any $u_1$ elements of $\{{\bf b}_1, {\bf b}_2, \ldots, {\bf b}_{n^{*}}\}$ are independent, we have
{\small
\begin{eqnarray}
 \mid \cap_{i=1}^v \overline E_i \mid =\mid E\setminus \cup_{i=1}^{v} E_{i}\mid &=&s^{u_1}-[{v\choose 1}s^{u_1-1} -{v\choose 2}s^{u_1-2} + \cdots +
                                                          (-1)^{v-1}{v\choose v} s^{u_1-v}]\nonumber\\
                                                      &=&s^{u_1}[1-{v\choose 1}s^{-1} + \cdots +(-1)^{v}{v\choose v}s^{-v}]\nonumber\\
                                                      &=&(s-1)^vs^{u_1-v}.\nonumber
\end{eqnarray}}
For $u_1+1\leq v\leq n^*$, the intersection of any $t\geq u_1$ sets of $E_i$'s only contains one vector, namely the zero column vector. Since
any $u_1$ elements of $\{{\bf b}_1, {\bf b}_2, \ldots, {\bf b}_{n^{*}}\}$ are independent, we have
\begin{eqnarray}
 \mid \cap_{i=1}^v \overline E_i \mid& = & \mid E\setminus \cup_{i=1}^{v} E_{i}\mid \\
 &=&s^{u_1}-[{v\choose 1}s^{u_1-1} -{v\choose 2}s^{u_1-2} + \cdots +
                                                         (-1)^{u_1-1}{v\choose u_1} s^{u_1-u_1}\nonumber\\
                                                      & &+ (-1)^{u_1}{v\choose u_1+1}\cdot 1 + \cdots + (-1)^{v-1}{v\choose v}\cdot 1]\nonumber\\
                                                      &=&s^{u_1}[1-{v\choose 1}s^{-1} + \cdots +(-1)^{u_1}{v\choose u_1}s^{-u_1}] + \sum_{i=u_1+1}^v(-1)^i{v\choose i}\nonumber\\
                                                      &=&m^*.\nonumber
\end{eqnarray}
\end{proof}

\noindent{\it Proof of  Theorem \ref{theorem:subspace-construction}}

\begin{proof}
Followed by Lemma \ref{lemma:Ei-and-Mi}, for any ${\bf z}\in \cap_{j=1}^v\overline E_{i_j}$ and
${\bf x}\in \cup_{j=1}^v\mathcal{A}_{i_j}$, we have ${\bf z}\notin O({\bf x})$.
Thus, by Lemmas \ref{lemma:basic-idea} and   \ref{lem:D1-D2-condition},
the $(D_1, D_2)$'s constructed in both items are marginally coupled designs.
In addition, both items $(i)$ and $(ii)$,
$O({\bf x}_i)\neq O({\bf x}_j)$ when $i\neq j$, which implies that ${\bf d}_i$ cannot be obtained
from ${\bf d}_j$ by level permutations. Therefore, $D_2$'s are non-cascading Latin hypercubes.
\end{proof}

\noindent{\it Proof of  Proposition \ref{prop:upperbound}}

\begin{proof}
Since any $u_1$ vectors of $\{{\bf b}_1, \ldots, {\bf b}_{n^*}\}$ are independent,
one can use them to obtain an $\OA(s^{u_1}, n^*, s, u_1)$. The run size here is $s^{u_1}$,
 not $s^u$, because the last $u-u_1$ entries of ${\bf b}_i$'s are zeros.
Note that the maximum value of $n^*$ must not be greater than the maximum value of $m$ for an $\OA(s^{u_1}, m, s, u_1)$ to exist. The right hand side of (\ref{eq:$n^*$})  are the upper bounds of $m$ for different cases, which were provided by
Theorem 2.19 of Hedayat, Sloane and Stufken (1999) .
\end{proof}

\noindent{\it Proof of  Proposition \ref{prop:anti-mirror}}

\begin{proof}
It is straightforward to see $\tilde D_2$ is an $\OA(2^u, 2^{u-u_1}, 2^{u-1}, 1)$. For $u-u_1>1$ and therefore $2^{u-u_1}>3$, consider a subarray $({\bf d}_p, {\bf d}_q, {\bf d}_l)$  of $\tilde D_2$, for $1\leq p< q< l\leq 2^{u-u_1}$.
Let ${\bf c}_i=\lfloor{\bf d}_{i}/2^{u-2}\rfloor$. As ${\bf d}_{i}=A_i\cdot(2^{u-2}, \ldots, 2, 1)^T$, ${\bf c}_i$ is the first column of $A_i$. In addition, $({\bf c}_p, {\bf c}_q, {\bf c}_l)$ is the projection of $({\bf d}_p, {\bf d}_q, {\bf d}_l)$
on the $2\times 2\times 2$ grid. Because $A_i$ is constructed by $G({\bf x}_i)$,  ${\bf c}_i$ is generated from $\eta_i$. As ${\bf y}_i\neq {\bf y}_j$ for $i\neq j$, we have $\overline {\bf y}_i\neq \overline {\bf y}_j$.
Since the last $u-u_1$ entries of $\eta_i$ is $\overline {\bf y}_i$,
$\eta_p, \eta_q$ and $\eta_l$ are three different columns. In addition, $\eta_{p}+\eta_q \neq \eta_l$
because the first $u_1$ entries
of $\eta_p, \eta_q, \eta_l$ are equal to ${\bf x}_0=(1,1,0, \ldots, 0)^T$. As a result, $\eta_{p}, \eta_q, \eta_l$ are three independent
column vectors.  Thus, the array $({\bf c}_p, {\bf c}_q, {\bf c}_l)$ is an $OA(2^u, 3, 2, 3)$, and
the conclusion follows.
\end{proof}

\noindent{\it Proof of  Proposition \ref{prop:5}}

\begin{proof}
In the subspace construction of Theorem \ref{theorem:subspace-construction}, for $i=1,\ldots,k$,
each $O({\bf x}_i)$ contains a set of $(s^{u-1}-1)/(s-1)$ different column vectors,
the first nonzero entry of each of which is equal to $1$.
If $k\leq (s^{u-1}-1)/(s-1)$, one can always choose ${\bf y}_i\in O({\bf x}_i)$,
such that ${\bf y}_i\neq \alpha {\bf y}_j$ for $1\leq i\neq j\leq k$ and any $\alpha\in GF(s)$.
Let ${\bf y}_i$ be the first column of $G({\bf x}_i)$ which is used to obtain $A_i$ and  consists of $u-1$ independent columns of $O({\bf x}_i)$.  For such $\{A_1, \ldots, A_k\}$,  the first $k$ columns  form an $\OA(s^u, k, s, 2)$, which guarantees
$\tilde D_2$ to achieve stratifications on an $s\times s$ grid of any two dimensions.
\end{proof}

\begin{proposition}\label{prop-the-bound-for-intersectionset}
The set $\cap_{i=1}^{n_B}\overline E_i$ is equal to $(i)$
$\{{\bf e}_{i_1} + {\bf e}_{i_2} + \cdots + {\bf e}_{i_{2t+1}} \mid 2t+1\leq u_1, 1\leq i_1<i_2<\ldots< i_{2t+1}\leq u_1\}$
when $s=2$, or equal to $(ii)$  $\{ \alpha{\bf e}_i \mid \alpha\in GF(s)\setminus\{0\}, i=1, \ldots, u_1 \}$ when $s>2$.
\end{proposition}
\begin{proof}
For $s=2$, we have $n_B=1$, $\mathcal{A}=\mathcal{A}_1$, and ${\bf b}_1=(1,\ldots, 1, 0, \ldots, 0)^T$ where the first $u_1$
entries are equal to 1. If ${\bf z}\in E$ and ${\bf z}^T{\bf b}_1\neq 0$,
${\bf z}$ must be a sum of an odd number of ${\bf e}_i$'s.
Thus, item (i) follows. If ${\bf z}\in \cap_{i=1}^{n_B}\overline E_i${\bf ,}
${\bf z}\notin O({\bf x})$ for any ${\bf x}\in \mathcal{A}$ by Lemma \ref{lemma:Ei-and-Mi}. Therefore, for $s>2$,
the possible elements in $\cap_{i=1}^{n_B}\overline E_i$ can only be ${\bf z}=\alpha{\bf e}_j$
for any $\alpha\in GF(s)\setminus\{0\}$ and $j=1, \ldots, u_1$,
according to Proposition \ref{thm-combination is impossible}, while ${\bf e}_j\in \cap_{i=1}^{n_B}\overline E_i$,
for $j=1,\ldots, u_1$. Combining these two results, item (ii) follows.
\end{proof}

\section*{Acknowledgements}

The authors wish to thank  the Editor, an Associate Editor, and two referees for their
helpful comments which have led to the improvement of the manuscript.

Yuanzhen He is supported by the National Natural Science
Foundation of China Grant 11701033. C. Devon Lin's research was supported by the Discovery grant from Natural Sciences and Engineering
Research Council of Canada. Fasheng Sun is supported by the National Natural Science
Foundation of China Grants 11471069,  11771220 and the Fundamental Research Funds for the Central Universities.

\vspace{0.3in}
\noindent {\bf \Large References}

\def\beginref{\begingroup
                \clubpenalty=10000
                \widowpenalty=10000
                \normalbaselines\parindent 0pt
                \parskip.0\baselineskip
                \everypar{\hangindent1em}}
\def\endref{\par\endgroup}
\renewcommand{\baselinestretch}{1}

\beginref

Deng, X., Hung, Y.\ and Lin, C.D.\ (2015).   Design for computer experiments with qualitative and quantitative factors.   {\em Statistica Sinica}, {\bf 25}, 1567--1581.

Deng, X., Lin, C.D., Liu, K.W.\ and Rowe, R.K. \ (2017). Additive Gaussian process for computer models with qualitative and quantitative factors.
{\em Technometrics}, {\bf 59}, 283--292.

Draguljic, D., Santner, T.J.\ and Dean, A.M.\  (2012). Noncollapsing space-filling designs for bounded nonrectangular regions.{\em Technometrics}, {\bf 54}, 169--178.

Han, G., Santner, T.J., Notz, W.I.\ and Bartel, D.L.\ (2009).  Prediction for computer experiments having quantitative  and qualitative input variables. {\em Technometrics}, { \bf 51}, 278--288.

Handcock, M.S.\ (1991). On cascading Latin hypercube designs and additive models for
experiments. {\em  Comm. Statist. Theory Methods},  {\bf 20}, 417--439.

He, Y.,  Lin, C.D.\ and Sun, F.S.\ (2017). On the construction of marginally coupled designs. {\em Statistica Sinica}, {\bf 27}, 665--683.

He, Y.,  Lin, C.D., Sun, F.S.\ and Lv, B.J.\ (2017). Marginally coupled designs for two-level qualitative factors. {\em Journal of Statistical Planning and Inference}, {\bf 187}, 103--108.


Hedayat, A.S., Sloane, N.J.A.\ and Stufken, J.\ (1999). {\it Orthogonal Arrays: Theory and Applications}. Springer, New York.

Horn, R.A.\ and Johnson, C.R.\ (2015). {\it Matrix analysis, Second Edition}. Cambridge University Press $\&$ Posts $\&$ Telecom Press.

Huang, H., Lin, D.K.J., Liu, M.Q.\ and Yang, J.F.\ (2016). Computer experiments with both qualitative and quantitative variables. {\em Technometrics}, {\bf 58}, 495--507.

Joseph, V.R., Gul, E.\ and Ba, S.\ (2015). Maximum projection designs for computer experiments. {\em Biometrika}, {\bf 102},371--380.

Leary, S., Bhaskar, A.\ and Keane, A.\ (2003). Optimal orthogonal-array-based Latin hypercubes. {\em Journal of Applied Statistics}, {\bf 30}, 585--598.


Lin, C.D.\ and Tang, B.\ (2015). Latin hypercubes and space-filling designs. In
{\em Handbook of Design and Analysis of Experiments}.  CRC Press. Bingham, D., Dean, A., Morris, M., and
Stufken, J. ed., 593--626.

McKay, M.D., Beckman, R.J.\ and Conover, W.J. (1979). A comparison of three methods for selecting values of input variables
in the analysis of output from a computer code.  {\em Technometrics}, {\bf 21}, 239--245.


Qian, P.Z.G.\ and Wu, C.F.J.\ (2009). Sliced space-filling designs. {\em Biometrika},
 {\bf 96}, 733--739.

Qian, P.Z.G., Wu, H.\ and Wu, C.F.J.\ (2008). Gaussian process models for computer experiments with qualitative and quantitative factors. {\em Technometrics},  {\bf 50}, 383--396.

Rawlinson, J.J., Furman, B.D., Li, S., Wright, T.M.\ and Bartel, D.L. \ (2006). Retrieval, experimental, and computational assessment of the performance of total knee replacements.
{\em Journal of Orthopaedic Research Official Publication of the Orthopaedic Research Society}, {\bf 24}, 1384--1394.


Sun, F.\  and Tang, B.\  (2017). A method of constructing space-filling orthogonal designs. {\em Journal of the American Statistical Association}, {\bf 112}, 683--689.

Tang, B.\ (1993). Orthogonal array-based Latin hypercubes.  {\em Journal of the American Statistical Association},  {\bf 88}, 1392--1397.

Wu, C.F.J.\ and Hamada, M.S. (2011). {\it Experiments: Planning, Analysis, and Optimization.} John Wiley \& Sons.

Xie, H., Xiong, S., Qian, P.Z.G.\ and Wu, C.FJ.\  (2014). General sliced Latin hypercube designs. {\em Statistica Sinica}, {\bf 24}, 1239--1256.

Zhou, Q., Qian, P.Z.G. and Zhou, S.\ (2011). A simple approach to emulation for computer models with qualitative and quantitative factors. {\em Technometrics}, {\bf 53}, 266--273.

Zhou, Q., Jin, T., Qian, P.Z.G.\ and Zhou, S. (2016). Bi-directional sliced Latin hypercube designs. {\em Statistica Sinica}, {\bf 26}, 653--674.

\endref

\end{document}